\documentclass[10pt]{llncs}
\usepackage[utf8]{inputenc}

\usepackage{amssymb}
\usepackage{amsfonts}
\usepackage{amsxtra}
\usepackage{graphicx}
\usepackage{url}
\usepackage{xcolor}

\title{A ``Symbolic'' Representation of Object-Nets
}
 \subtitle{ (Extended Version)}
\author{
  Michael K\"ohler-Bu\ss{}meier\orcidID{0000-0002-3074-4145}
  \inst{1}
  \and \\
  Lorenzo Capra\orcidID{0000-0002-1029-1169}%
  \inst{2}
}
\institute{%
  University of Applied Science Hamburg 
  \\  Berliner Tor 7, D-20099 Hamburg, Germany
  \\  \email{michael.koehler-bussmeier@haw-hamburg.de} 
  \and
  Dipartimento di Informatica, Universit\`a  degli Studi di Milano
  \\  Via Celoria 18, Milan, Italy 
  \\  \email{capra@di.unimi.it}
}

\newcommand{\EOS}{\textsc{Eos}}

\newcommand{\ms}[1]{\mathbf{#1}}

\newcommand{\labtr}[2][{}]{\xrightarrow[#1]{#2}}
\newcommand{\labtra}[2]{\xrightarrow[#2]{#1}}

\newcommand{\Nat}{\mathbb{N}}

\newcommand{\sn}[1]{\widehat{#1}} 
\newcommand{\on}[1]{#1}
\newcommand{\postset}[1]{{{#1}^\bullet}}
\newcommand{\preset}[1]{{{}^\bullet #1}}

\newcommand{\pre}{\mathbf{pre}} 
\newcommand{\post}{\mathbf{post}}

\newcommand{\Nets}{\mathcal{N}} 
\newcommand{\impl}{\Longrightarrow}
\newcommand{\key}[1]{\emph{#1}}

\newcommand{\morph}{\varphi}

\newcommand\singleFigWide[3][1.0]{%
  \begin{figure}[htbp]
    \centerline{\includegraphics[width=#1\textwidth]{bilder/#2}}
    \caption{\label{fig:#2}#3}
  \end{figure}%
}

\begin{document}

\maketitle

\begin{abstract}
  In this contribution we extend the concept of a Petri net morphism
  to Elementary Object Systems (\EOS{}).  \EOS{} are a
  nets-within-nets formalism, i.e. we allow the tokens of a Petri net
  to be Petri nets again.
  This nested structure has the consequence that even systems defined
  by very small Petri nets have a quite huge reachability graph.
  In this contribution we use automorphism to describe symmetries of
  the Petri net topology.  Since these symmetries carry over to
  markings as well this leads to a condensed state space, too.

\end{abstract}

\section{Exploiting Symmetry and Canonical Representations}

In this paper we study Elementary Object Systems (\EOS{})
\cite{Koehler14-fi-lam} a Nets-within-Nets formalism as proposed by
Valk \cite{Valk-acpn}, i.e., we allow the tokens of a Petri net to be
Petri nets again.  Due to the nesting structure many of the classical
decision problems, like reachability and liveness, become undecidable
for \EOS{}.

From a complexity perspective we have studied these problems for
\emph{safe \textsc{Eos}}
\cite{koehler+10-fi,Koehler+11-fi,Koehler14-fi-lam} where markings are
restricted to \emph{sets} (i.e., places are either marked or
unmarked).  More precisely: All problems that are expressible in LTL
or CTL, which includes reachability and liveness, are
\textsc{PSpace}-complete.  This means that in terms of complexity
theory {safe \textsc{Eos}} are no more complex than safe place
transition nets (p/t nets).
But, a look at the details shows a difference that is pratically
relevant: For safe p/t nets it is known that whenever there are $n$
places, then the number of reachable states is bounded by $O(2^{n})$;
but, for safe \EOS{} the number of reachable states is in
$O(2^{(n^2)})$ -- a quite drastic increase.
Therefore, our main goal is to derive a \emph{condensed} state space
for \EOS{}, were `condensed' is expressed as a factorisation modulo an
equivalence.

In this contribution we extend the concept of a Petri net morphism to
Elementary Object Systems (\EOS{}).  \EOS{} are a nets-within-nets
formalism.  Here, we use automorphism to describe symmetries of the
Petri net topology.  Since these symmetries carry over to markings as
well this leads to a condensed state space, too.  In our approach
these symmetries are introduced very naturally to the representation
of the state space using \emph{canonical representations} of
markings.

The paper has the following structure.
Section~\ref{sec:nin} introduces base nets-within-nets (\EOS{}).
In Section~\ref{sec:sym} we 
define a symbolic representation
of the \EOS{} structure.
The work closes with a conclusion and outlook.

\section{Nets-within-Nets, EOS}
\label{sec:nin}

Object nets \cite{Koehler+04-atpn,koehler+09-fi,Koehler14-fi-lam},
follow the \emph{nets-within-nets} paradigm as proposed by Valk
\cite{Valk-acpn}.
Other approaches adapting the nets-within-nets approach are nested
nets \cite{Lomazova00}, mobile predicate/transition nets \cite{xu+00},
$\mathrm{PN}^2$ \cite{Hiraishi02}, hypernets \cite{hypernets04}, and
adaptive workflow nets \cite{vanHee+06-adaptive-wf}. There are
relationships to Rewritable Petri nets \cite{lcEPTCS21} and
Reconfigurable Petri Nets \cite{Padberg-Kahloul-2018}.
Object Nets can be seen as the Petri net perspective on contextual
change, in contrast to the Ambient Calculus \cite{Cardelli+99b} or the
$\pi$-calculus \cite{Milner92a}, which form the process algebra perspective.

\subsection{Petri Nets}
\label{sec:petrinets}

The definition of Petri nets relies on the notion of multisets.  A multiset
$\ms{m}$ on the set $D$ is a mapping $\ms{m}: D \to \Nat$.  
Multisets are generalisations of sets in the sense
that every subset of $D$ corresponds to a multiset $\ms{m}$ with
$\ms{m}(d) \leq 1$ for all $d \in D$.
The empty multiset $\mathbf{0}$ is defined as $\mathbf{0}(d)
= 0$ for all $d \in D$.  
Multiset addition for $\ms{m}_1, \ms{m}_2: D \to \Nat$ is defined
component-wise: $(\ms{m}_1 + \ms{m}_2)(d) := \ms{m}_1(d) +
\ms{m}_2(d)$.
Multiset-difference $\ms{m}_1 - \ms{m}_2$ is defined by $(\ms{m}_1 -
\ms{m}_2)(d) := \max(\ms{m}_1(d) - \ms{m}_2(d), 0)$.
We use common notations for the cardinality of a multiset $|\ms{m}| :=
\sum_{d \in D} \ms{m}(d)$ and multiset ordering $\ms{m}_1 \leq
\ms{m}_2$, where the partial order $\leq$ is defined by $\ms{m}_1 \leq
\ms{m}_2 \iff \forall d \in D: \ms{m}_1(d) \leq \ms{m}_2(d)$.

A multiset $\ms{m}$ is finite if $|\ms{m}| < \infty$.  The set of all
finite multisets over the set $D$ is denoted $\mathit{MS}(D)$.
The set $\mathit{MS}(D)$ naturally forms a monoid with multiset
addition $+$ and the empty multiset $\mathbf{0}$.  Multisets can be
identified with the commutative monoid structure $(\mathit{MS}(D),
+, 0)$.  Multisets are the free commutative monoid over $D$ since every
multiset has the unique representation in the form $\ms{m} = \sum_{d
  \in D} \ms{m}(d) \cdot d$, where $\ms{m}(d)$ denotes the multiplicity
of $d$.  Multisets can also be represented as a formal sum in the form
$\ms{m} = \sum_{i=1}^{n} x_i$, where $x_i \in D$.

Any mapping $f: D \to D'$ is  extended to a multiset homomorphism $f^\sharp:
\mathit{MS}(D) \to \mathit{MS}(D')$: \(
f^\sharp\left(\sum_{i=1}^{n} x_i\right) = \sum_{i=1}^n f(x_i) \).
This includes the special case $f^\sharp(\mathbf{0}) = \mathbf{0}$.  We simply
write $f$ to denote the mapping $f^\sharp$. The notation is in
accordance with the set-theoretic notation $f(A) = \{ f(a) \mid a \in
A \}$.

\begin{definition}  
  A \key{p/t net} $N$ is a tuple \( N = (P, T, \pre,\post), \) such
  that $P$ is a set of places, $T$ is a set of transitions, with $P
  \cap T = \emptyset$, and $\pre, \post: T \to \mathit{MS}(P)$ are the
  pre- and post-condition functions.  A marking of $N$ is a multiset
  of places: $\ms{m} \in \mathit{MS}(P)$.  A {p/t net} with initial
  marking $\ms{m}_0$ is denoted $N = (P, T, \pre,\post, \ms{m}_0)$.
\end{definition}
We use the usual notation for nets such as $\preset{x}$ for the set of
predecessors and $\postset{x}$ for the set of successors for a node $x
\in (P \cup T)$.
For $t\in T$ we have
$\preset{t} = \{ p \in P \mid \pre(t)(p)>0 \}$ and
$\postset{t} = \{ p \in P \mid \post(t)(p)>0 \}$.
For $p\in P$ we have
$\preset{p} = \{ t \in T \mid \post(t)(p)>0 \}$ and
$\postset{p} = \{ t \in T \mid \pre(t)(p)>0 \}$.

A transition $t \in T$ of a p/t net $N$ is enabled in marking $\ms{m}$
iff $\forall p \in P: \ms{m}(p) \geq \pre(t)(p)$ holds.  The successor
marking when firing $t$ is $\ms{m}'(p) = \ms{m}(p) - \pre(t)(p) +
\post(t)(p)$ for all $p \in P$.
Using multiset notation enabling is expressed by $\ms{m} \geq \pre(t)$
and the successor marking is $\ms{m}' = \ms{m} - \pre(t) + \post(t)$.
We denote the enabling of $t$ in marking
$\ms{m}$ by $\ms{m} \labtra{t}{N}$. Firing of an enabled $t$ is denoted by
$\ms{m}\labtra{t}{N} \ms{m}'$.
The net $N$ is omitted if it is clear from the context.

Firing is extended to sequences $w \in T^*$ in the obvious way: 
(i) $\ms{m} \labtr{\epsilon} \ms{m}$; 
(ii) If $\ms{m} \labtr{w} \ms{m}'$ and $\ms{m}' \labtr{t} \ms{m}''$
hold, then we have $\ms{m}\labtr{w t} \ms{m}''$. 
We write $\ms{m}\labtr{*} \ms{m}'$ whenever there is some
$w \in T^*$ such that $\ms{m}\labtr{w} \ms{m}'$ holds.
The set of reachable markings is $\mathit{RS}(\ms{m}_0) := \{ \ms{m}
\mid \exists w \in T^*: \ms{m}_0 \labtr{w} \ms{m} \}$.

\subsection{Elementary Object Systems}
\label{sec:eos}

In the following we consider \key{Elementary Object System} (\EOS{})
\cite{Koehler14-fi-lam}, which have a two-levelled structure.
An elementary object system (\textsc{Eos}) is composed of a system
net, which is a p/t net $\sn{N} = (\sn{P}, \sn{T}, \pre, \post)$ and a
set of object nets $\Nets = \{\on{N}_{1}, \ldots, \on{N}_{n}\}$, which
are p/t nets given as $\on{N} = (\on{P}_{\on{N}}, \on{T}_{\on{N}},
\pre_{\on{N}}, \post_{\on{N}})$, where $\on{N} \in \Nets$.  In
extension we assume that all sets of nodes (places and transitions)
are pairwise disjoint.  Moreover we assume $\sn{N} \not\in \Nets$ and
the existence of the object net $\bullet \in \Nets$, which
has no places and no transitions and is used to model anonymous, so
called black tokens.

\paragraph{Typing}

The system net places are typed by the mapping $d: \sn{P} \to \Nets$
with the meaning, that a place $\sn{p} \in \sn{P}$ of the system net
with $d(\sn{p}) = \on{N}$ may contain only net-tokens of the object
net type $\on{N}$.\footnote{In some sense, net-tokens are object nets
  with their own marking.
  However, net-tokens should not be considered as \emph{instances} of an
  object net (as in object-oriented programming), since net-tokens do
  not have an identity.  This is reflected by the fact that the firing
  rule joins and distributes the net-tokens' markings.

  Instead, all net-tokes of an object net act as a \emph{collective}
  entity, like a group.  This group can be considered as an object
  with identy -- an object with its state distributed over the
  net-tokens.
  For in in-depth discussion of this semantics cf.~\cite{Valk-acpn}. }
No place of the system net is mapped to the system net itself since
$\sn{N} \not\in \Nets$.

A typing is called \key{conservative} iff for each place $\sn{p}$ in the
preset of a system net transition $\sn{t}$ such that $d(\sn{p})
\not= \bullet$ there is place in the postset being of the same type:
$(d(\preset{\sn{t}}) \cup \{\bullet\}) \;\subseteq\;
(d(\postset{\sn{t}}) \cup \{\bullet\})$.
An \textsc{Eos} is \key{conservative} iff its typing $d$ is.

An \textsc{Eos} is \key{p/t-like} iff it has only places for
black tokens: $d(\sn{P})= \{\bullet\}$.

\paragraph{Nested Markings}

Since the tokens of an \textsc{Eos} are instances of object nets, a
\key{marking} of an \textsc{Eos} is a \emph{nested} multiset.  A
marking of a \textsc{Eos} $\mathit{OS}$ is denoted $\mu =
\sum_{k=1}^{|\mu|} (\sn{p}_k, \on{M}_k)$, where $\sn{p}_k$ is a place
of the system net and $\on{M}_k$ is the marking of a net-token of
type $d(\sn{p}_k)$.  To emphasize nesting, the marks are also
denoted as $\mu = \sum_{k=1}^{|\mu|} \sn{p}_k[\on{M}_k]$.  Markings of
the form $\sn{p}[\mathbf{0}]$ with $d(\sn{p}) = \bullet$ are
abbreviated as $\sn{p}[]$.

The set of all markings which are syntactically consistent with the
typing $d$ is denoted $\mathcal{M}$, where $d^{-1}(\on{N}) \subseteq
\sn{P}$ is the set of system net places of the type $\on{N}$:
\begin{equation}
  \label{eq:nestedmultisets}
  \mathcal{M}
  := \mathit{MS}\left(
    \bigcup\nolimits_{\on{N} \in \Nets} 
    \left(d^{-1}(\on{N}) \times \mathit{MS}(\on{P}_{\on{N}}) \right)\right)  
\end{equation}

We define the partial order $\sqsubseteq$ on nested multisets by setting
$\mu_1 \sqsubseteq \mu_2$ iff $\exists \mu: \mu_2 = \mu_1 + \mu$.
A more liberal variant is the order $\preceq$ defined by:
\begin{equation}
  \label{eq:poset-nestedmultisets}
  \begin{array}{rl}
    \alpha \preceq \beta \quad\iff & 
                                     \alpha = \sum_{i=1}^{m} \sn{a}_i[\on{A}_i]
                                     \land
                                     \beta = \sum_{j=1}^{n} \sn{b}_j[\on{B}_j]
                                     \land {}
    \\&
    \exists \text{ injection } f: \{1,...,m\} \to \{1,...,n\}:
    \\&
    \forall 1 \leq i \leq m:
    \sn{a}_i = \sn{b}_{f(i)} \land
    \on{A}_i \leq  \on{B}_{f(i)}
  \end{array}
\end{equation}  

For $ \alpha \preceq \beta$ the injection $f$ generates a sub-marking
of $\beta$ which is denoted $f(\alpha) = \sum_{i=1}^{m}
\sn{a}_{f(i)}[\on{A}_{f(i)}]$.
Note that $\alpha \sqsubseteq \beta$ is a special case of $\alpha
\preceq \beta$, where $\on{A}_i \leq \on{B}_{f(i)}$ is restricted to
$\on{A}_i = \on{B}_{f(i)}$.

\paragraph{Events}

Analogously to markings, which are nested multisets $\mu$, the events
of an \textsc{Eos} are also nested.  An \textsc{Eos} allows three
different kinds of events -- as illustrated by the 
\textsc{Eos} in Fig.~\ref{fig:eos}.

\begin{figure}[htbp]
  \centerline{\includegraphics[width=0.99\textwidth]{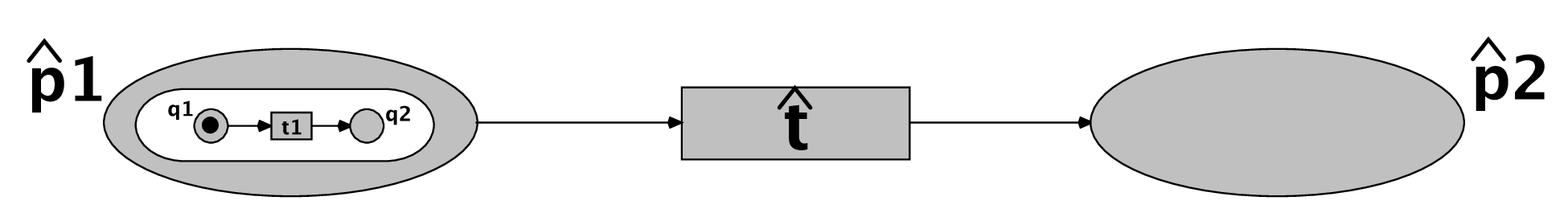}}
  \caption{\label{fig:eos} An Elementary Object Net System (\EOS)}
\end{figure}%

\begin{enumerate}
\item System-autonomous: The system net transition $\sn{t}$ fires
  autonomously which moves the net-token from $\sn{p}_1$ to $\sn{p}_2$
  without changing its marking.

\item Object-autonomous: The object net fires transition $\on{t}_1$,
  which ``moves'' the black token from $\on{q}_1$ to $\on{q}_2$.  The object
  net itself remains at its location $\sn{p}_1$.
\item Synchronisation: The system net transition $\sn{t}$ fires
  synchronously with $\on{t}_1$ in the object net.  Whenever
  synchronisation is demanded, autonomous actions are forbidden.
\end{enumerate}

The  set of events is denoted $\Theta$.
Events are formalised as a pair $\sn{\tau}[\vartheta]$, where
$\sn{\tau}$ is either the transition that fires in the system net or a
special ``idle'' transition $\mathit{id}_{\sn{p}}$ (cf. below); and
$\vartheta$ is a function such that $\vartheta(\on{N})$ is the
multiset of transitions, which have to fire synchronously with
$\sn{\tau}$, (i.e.\ for each object net $\on{N} \in \Nets$ we have
$\vartheta(\on{N}) \in \mathit{MS}(\on{T}_{\on{N}})$).\footnote{ In
  the graphical representation the events are generated by transition
  inscriptions.  For each object net $\on{N} \in \Nets$ a system net
  transition $\sn{t}$ is labelled with a multiset of channels
  $\sn{l}(\sn{t})(\on{N}) = \mathit{ch}_1 + \cdots + \mathit{ch}_n $,
  depicted as $\langle \on{N}\textit{:ch}_1, \on{N}\textit{:ch}_2,
  \ldots \rangle$.
  Similarily, an object net transition $\on{t}$ may be labelled with a
  channel $\on{l}_{\on{N}} (\on{t}) = \mathit{ch}$ -- depicted as
  $\langle \textit{:ch} \rangle$ whenever there is such a label.
  We obtain an event $\sn{t}[\vartheta]$ by setting $\vartheta(\on{N})
  := \on{t}_1 + \cdots + \on{t}_n$ to be any transition multiset such
  that the labels match: \( \on{l}_{\on{N}} (\on{t}_1) + \cdots +
  \on{l}_{\on{N}} (\on{t}_n) = \sn{l}(\sn{t})(\on{N}) \).  }

In general $\sn{\tau}[\vartheta]$ describes a synchronisation, but
autonomous events are special subcases:
Obviously, a system-autonomous event is the special case, where
$\vartheta = \mathbf{0}$ with $\mathbf{0}(\on{N}) = \mathbf{0}$ for
all object nets $\on{N}$.
To describe an object-autonomous event we assume the set of \emph{idle
  transitions} $\{ \mathit{id}_{\sn{p}} \mid \sn{p} \in \sn{P} \}$,
where $\mathit{id}_{\sn{p}}$ formalises object-autonomous firing on
the place $\sn{p}$:
\begin{enumerate}
\item Each idle transition $\mathit{id}_{\sn{p}}$ has $\sn{p}$ as a
  side condition: $\pre(\mathit{id}_{\sn{p}}) =
  \post(\mathit{id}_{\sn{p}}) := \sn{p}$.

\item Each idle transition $\mathit{id}_{\sn{p}}$ synchronises only
  with transitions from $\on{N} = d(\sn{p})$:
  \begin{equation*}
    \begin{array}{rl}
      \forall \sn{\tau}[\vartheta]  \in    \Theta:
      \; \sn{\tau} = \mathit{id}_{\sn{p}}
      \;\;\;\impl\;\;\;
      \forall \on{N} \in \Nets: 
        &
          (
          \vartheta(\on{N}) \not= \mathbf{0}
          \iff
          \on{N} = d(\sn{p})
          )
    \end{array}
  \end{equation*}
\end{enumerate}

\begin{definition}[Elementary Object System, EOS]
  \label{def:EOS}
  An elementary object system (\textsc{Eos}) is a tuple 
  \(
  \mathit{OS} = (\sn{N}, \Nets, d, \Theta),  
  \)
  where:
  \begin{enumerate}
  \item $\sn{N}$ is a p/t net, called the \key{system net}.   
  \item $\Nets$ is a finite set of disjoint p/t nets, called
    \key{object nets}.
  \item $d: \sn{P} \to \Nets$ is the \key{typing} of the system net places.
  \item $\Theta$ is the set of \key{events}.
  \end{enumerate}

  An \textsc{Eos} $\mathit{OS}$ with initial marking 
  $ \mu_{0}$ is a marked \EOS{}.
  We use the term \EOS{} both for marked and unmarked systems.
\end{definition}

\begin{example}
  Figure~\ref{fig:eos-s8} shows an \textsc{Eos} with the system
  net $\sn{N}$ and the object nets $\Nets = \{\on{N}_1, \on{N}_2\}$.
  The system has four net-tokens: two on place $\sn{p}_{1}$ and one on
  $\sn{p}_{2}$ and $\sn{p}_{3}$ each.  
  (Please ignore the net-tokens above the transition
  and on the places
  $\sn{p}_{4}$,
  $\sn{p}_{5}$ and $\sn{p}_{6}$ on the right; they are use below to illustrate the firing rule.) 
  The net-tokens on $\sn{p}_{1}$ and $\sn{p}_{2}$ share the same net
  structure, but have independent markings.

  \singleFigWide{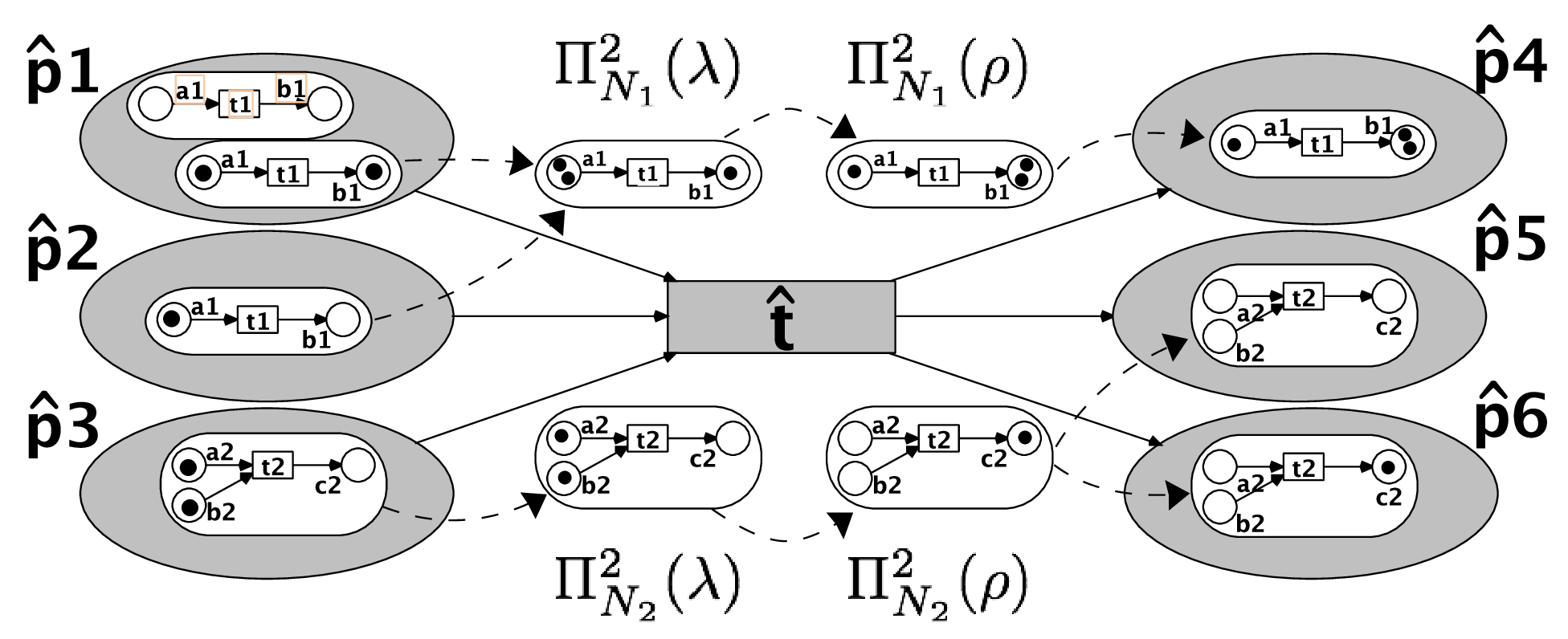}{An Example  \textsc{EOS} 
    illustrating the firing of $\sn{t}[\on{N}_1 \mapsto t_1,\on{N}_2  \mapsto t_2]$ 
  }

  \noindent The system net is $\sn{N}= (\sn{P}, \sn{T}, \pre, \post)$,
  where $\sn{P} = \{\sn{p}_1,\ldots,\sn{p}_6\}$ and $\sn{T} = \{\sn{t}\}$.

  \noindent One object net is  $\on{N}_1= (\on{P}_1, \on{T}_1, \pre_1,
  \post_1)$ with $\on{P}_1 = \{a_1,b_1\}$ and $\on{T}_1 = \{t_1\}$.

  \noindent Another  object net is $\on{N}_2= (\on{P}_2, \on{T}_2, \pre_2,
  \post_2)$ with $\on{P}_2 = \{a_2,b_2,c_2\}$ and $\on{T}_2 =
  \{t_2\}$.

  \noindent The typing is $d(\sn{p}_1) = d(\sn{p}_2) = d(\sn{p}_4) = \on{N}_1$ and $d(\sn{p}_3)
  = d(\sn{p}_5) = d(\sn{p}_6) = \on{N}_2$.

  \noindent We have only one event: $\Theta = \{ \sn{t}[
  \on{N}_1 \mapsto t_1, \on{N}_2 \mapsto t_2] \}$.

  \noindent The initial marking has two net-tokens on $\sn{p}_1$, one
  on $\sn{p}_2$, and one on $\sn{p}_3$:
  \[
    \mu  = \sn{p}_1[a_1 + b_1] + \sn{p}_1[\mathbf{0}] + \sn{p}_2[a_1]  + \sn{p}_3[a_2 + b_2]
  \]

  Note  that for Figure~\ref{fig:eos-s8} the structure is the same
  for the three net-tokens on $\sn{p}_1$ and $\sn{p}_2$ but the net-token
  markings are different.
\end{example}

\paragraph{Firing Rule}

The projection $\Pi^1$ on the first component abstracts from the
substructure of all net-tokens for a marking of an \textsc{Eos}:
\begin{equation}
  \label{eq:1}
  \Pi^1\left(\sum\nolimits_{k=1}^{n}
    \sn{p}_k[ \on{M}_k]   \right)
  := \sum\nolimits_{k=1}^{n} \sn{p}_k  
\end{equation}

The projection $\Pi^2_{\on{N}}$ on the second component is the sum of
all net-token markings $\on{M}_k$ of the type $\on{N} \in \Nets$,
ignoring their local distribution within the system net:
\begin{equation}
  \label{eq:2}
  \Pi^2_{\on{N}} \left(\sum\nolimits_{k=1}^{n}  
    \sn{p}_k[\on{M}_k]  \right)
  := \sum\nolimits_{k=1}^{n}
  \mathbf{1}_{\on{N}}(\sn{p}_k) \cdot \on{M}_k  
\end{equation}
where the indicator function $\mathbf{1}_{\on{N}}: \sn{P} \to \{0,1\}$
is   $\mathbf{1}_{\on{N}}(\sn{p}) = 1$ iff $d(\sn{p}) = \on{N}$.
Note  that $\Pi^2_{\on{N}}(\mu)$ results in a marking of the object
net $\on{N}$.

A system event $\sn{\tau}[\vartheta]$ removes net-tokens together with
their individual internal markings.  Firing the event replaces a
nested multiset $\lambda \in \mathcal{M}$ that is part of the current
marking $\mu$, i.e.\ $\lambda \sqsubseteq \mu$, by the nested multiset
$\rho$. Therefore the successor marking is $\mu' := (\mu - \lambda) +
\rho$.
The enabling condition is expressed by the \key{enabling predicate}
$\phi_{\mathit{OS}}$ (or just $\phi$ whenever $\mathit{OS}$ is clear
from the context):

\begin{equation}
  \label{eq:firepredicate}  
  \begin{array}{rl}
    \phi(\sn{\tau}[\vartheta], \lambda, \rho)
    &\iff
      \Pi^1(\lambda) = \pre(\sn{\tau}) \land
      \Pi^1(\rho) = \post(\sn{\tau}) \land{}
    \\ &
         \forall\on{N} \in \Nets: \Pi^2_{\on{N}}(\lambda)
         \geq \pre_{\on{N}}(\vartheta(\on{N})) \land{}
    \\ &
         \forall\on{N} \in \Nets:
         \Pi^2_{\on{N}}(\rho) = \Pi^2_{\on{N}}(\lambda) 
         - \pre_{\on{N}}(\vartheta(\on{N})) 
         + \post_{\on{N}}(\vartheta(\on{N})) 
  \end{array}
\end{equation}

With $\sn{M} = \Pi^1(\lambda)$ and $\sn{M}' = \Pi^1(\rho)$ as well as
$\on{M}_{\on{N}} = \Pi^2_{\on{N}}(\lambda)$ and $\on{M}'_{\on{N}} =
\Pi^2_{\on{N}}(\rho) $ for all $\on{N} \in \Nets$ the predicate $\phi$
has the following meaning:
\begin{enumerate}
\item The first conjunct expresses that the system net multiset
  $\sn{M}$ corresponds to the pre-condition of the system net
  transition $\sn{\tau}$, i.e.~$\sn{M} = \pre(\sn{\tau})$.
  
\item In turn, a multiset $\sn{M}'$ is produced, that corresponds to
  the post-set of $\sn{\tau}$.  
  
\item A multi-set $\vartheta(\on{N})$ of object net transitions is
  enabled if the sum $\on{M}_{\on{N}}$ of the net-token markings (of
  type $\on{N}$) enable it, i.e.\ $\on{M}_{\on{N}}\geq
  \pre_{\on{N}}(\vartheta(\on{N}))$.
  
\item The firing of $\sn{\tau}[\vartheta]$ must also obey the
  \key{object marking distribution condition}: $ \on{M}'_{\on{N}} =
  \on{M}_{\on{N}} - \pre_{\on{N}}(\vartheta(\on{N})) +
  \post_{\on{N}}(\vartheta(\on{N}))$, where
  $\post_{\on{N}}(\vartheta(\on{N})) -
  \pre_{\on{N}}(\vartheta(\on{N}))$ is the effect of the object net's
  transitions on the net-tokens.

\end{enumerate}

Note that conditions 1.\ and~2.\ assure that only net-tokens relevant
for the firing are included in $\lambda$ and $\rho$.  Conditions 3.\
and~4.\ allow for additional tokens in the net-tokens.

For system-autonomous events $\sn{t}[\mathbf{0}]$ the enabling
predicate $\phi$ can be simplified further.  We have
$\pre_{\on{N}}(\mathbf{0}(\on{N})) =
\post_{\on{N}}(\mathbf{0}(\on{N})) = \mathbf{0}$.  This ensures
$\Pi^2_{\on{N}}(\lambda) = \Pi^2_{\on{N}}(\rho)$, i.e.\ the sum of
markings in the copies of a net-token is preserved w.r.t.\ each type
$\on{N}$.  This condition ensures the existence of linear invariance
properties

Analogously, for an object-autonomous event we have an idle-transition
$\sn{\tau} = \mathit{id}_{\sn{p}}$ for the system net and the first
and the second conjunct is: $\Pi^1(\lambda) =
\pre(\mathit{id}_{\sn{p}}) = \sn{p} = \post(\mathit{id}_{\sn{p}}) =
\Pi^1(\rho)$.  So, there is an addend $\lambda = \sn{p}[ \on{M}]$ in
$\mu$ with $d(\sn{p})= \on{N}$ and $\on{M}$ enables
$\vartheta(\on{N})$.

\begin{definition}[Firing Rule]
  \label{def:eos-fire}
  Let $\mathit{OS}$ be an \textsc{Eos} and $\mu, \mu' \in \mathcal{M}$
  markings.  The event $\sn{\tau}[\vartheta]$ is enabled in $\mu$ for the
  mode $(\lambda, \rho) \in \mathcal{M}^{2}$ iff $\lambda \sqsubseteq \mu
  \land {\phi}(\sn{\tau}[\vartheta], \lambda, \rho)$ holds.

  An event $\sn{\tau}[\vartheta]$ that is enabled in $\mu$ for the mode
  $(\lambda, \rho)$ can fire: \( \mu
  \labtra{\sn{\tau}[\vartheta](\lambda, \rho)}{\mathit{OS}} \mu' \).  The
  resulting successor marking is defined as $\mu' = \mu - \lambda +
  \rho$.
\end{definition}

We write $\mu \labtra{\sn{\tau}[\vartheta]}{\mathit{OS}} \mu'$
whenever $\mu \labtra{\sn{\tau}[\vartheta](\lambda,
  \rho)}{\mathit{OS}} \mu'$ for some mode $(\lambda,
\rho)$.

Note that the firing rule makes no a-priori assumptions about how to
distribute the object net markings onto the generated
net-tokens. Therefore we need the mode $(\lambda, \rho)$ to formulate
the firing of $\sn{\tau}[\vartheta]$ in a functional way.

\begin{example}
  Consider the \textsc{Eos} of Figure~\ref{fig:eos-s8} again.  The
  current marking $\mu$ of the \textsc{Eos} enables $\sn{t}[\on{N}_1
  \mapsto t_1,\on{N}_2 \mapsto t_2]$ in the mode $(\lambda, \rho)$,
  where
  \[
    \begin{array}{rcl}
      \mu  &=& \sn{p}_1[ \mathbf{0}]
               + {\sn{p}_1[ a_1 + b_1] +  \sn{p}_2[ a_1]  + \sn{p}_3[ a_2 + b_2]}
               = \sn{p}_1[\mathbf{0}]     + \lambda
      \\
      \lambda  &=& \sn{p}_1[a_1 + b_1]  + \sn{p}_2[ a_1]  + \sn{p}_3[ a_2 + b_2]
      \\
      \rho &=&  \sn{p}_4[ a_1 + b_1 + b_1] + \sn{p}_5[ \mathbf{0}]  + \sn{p}_6[c_2]
    \end{array}
  \]

  The net-token markings are added by the projections
  $\Pi^2_{\on{N}}$ resulting in the markings
  $\Pi^2_{\on{N}}(\lambda)$.  The sub-synchronisation generates
  $\Pi^2_{\on{N}}(\rho)$. (The results are shown above and below the
  transition $\sn{t}$.)  After the synchronisation we obtain the successor
  marking $\mu'$ with net-tokens on $\sn{p}_4$, $\sn{p}_5$, 
  and $\sn{p}_6$ as shown
  in Figure~\ref{fig:eos-s8}:
  \[
    \begin{array}{rcl}
      \mu' &=&    (\mu - \lambda) + \rho
               = \sn{p}_1[ \mathbf{0}] + \rho
      \\
           &=&  \sn{p}_1[ \mathbf{0}]
               +  \sn{p}_4[ a_1 + b_1 + b_1] + \sn{p}_5[ \mathbf{0}]  + \sn{p}_6[  c_2]
    \end{array}
  \]
  Note, that we have only presented one mode $(\lambda, \rho)$ and
  that other modes are possible, too.
\end{example}

\section{\EOS{}-Automorphism and  Canonical Representation}
\label{sec:sym}

The pseudo-symbolic representation of \EOS{} that we will introduce
relies on established concepts like graph (auto)morphism and canonical
representative. In what follows, we implicitly employ multiset
homomorphism.

\begin{definition}
  \label{def:PT-morph}
  Given \key{p/t nets} $N$ and $N'$, a morphism between $N$ and $N'$
  is a pair $\morph=(\morph_t, \morph_p)$ of bijective maps $\morph_t
  : T_N \rightarrow T_{N'}, \morph_p : P_N \rightarrow P_{N'}$ such
  that:
  $$
  \forall t \in T_N : \ \pre_{N'}(\morph_t(t)) = \morph_p(\pre_N(t))
  \wedge \post_{N'}(\morph_t(t)) = \morph_p(\post_N(t))
  $$
  \noindent We use the notation $\morph : N \rightarrow N'$, and
  $N \cong N'$ means there exists $\morph : N \rightarrow N'$.

  \noindent A morphism between marked \key{p/t nets} $(N, \ms{m})$ and
  $(N', \ms{m'})$ is a \key{p/t} morphism
  $\morph : N \rightarrow N'$
  such that $\morph_p(\ms{m}) = \ms{m}'$.\\
  \noindent (We extend the notation introduced above to marked nets' morphism.)

  \medskip
  \noindent A morphism $\morph : N \rightarrow N$ is referred to as an
  automorphism: $\morph_p, \morph_t$ are permutations. The markings
  $\ms{m}_1$ and $\ms{m}_2$ of $N$ are said equivalent if and only if
  there exists a morphism
  $\morph : (N, \ms{m}_1) \rightarrow (N, \ms{m}_2)$. We denote this
  by $\ms{m}_1 \cong \ms{m}_2$. Furthermore, $\cong$ establishes an
  equivalence relation on the set of markings of $N$.

\end{definition}

A \key{p/t} morphism maintains the firing rule. Let $\morph : N \rightarrow N'$. 
\begin{equation}
  \ms{m}_1\labtra{t}{N} \ms{m}_2 \impl 
  \morph_p(\ms{m}_1)\labtra{\morph_t(t)}{N'} \morph_p(\ms{m}_2)
\end{equation}
The equivalence relation $\cong$ on the markings of a \key{p/t net}
$N$ is thus a congruence when considering transition firing.

\subsection{\EOS{}-Automorphism}

We provide a natural extension of automorphisms to \EOS{}, which
considers that \EOS{} components are {p/t nets}.

\begin{definition}
  \label{def:eos-automorph}
  Let \(\mathit{OS} = (\sn{N}, \Nets, d, \Theta)\) be an \EOS{}  (Def. \ref{def:EOS}). 

  An \key{\EOS{}-automorphism $\morph_{OS}$} 
  is a collection of \key{p/t} automorphisms
  \[
    \morph_{OS} = \big(\morph_{\sn{N}}: \sn{N} \to \sn{N},\,\, (\morph_{\on{N}} : \on{N} \to \on{N})_{N \in \Nets} \big)
  \]
  that preserve $d$ and $\Theta$:      
  \begin{enumerate}
  \item $\forall \sn{p} \in P_{\sn{N}}: \ d(\morph_{\sn{N}}(\sn{p})) = d(\sn{p})$
    \smallskip

  \item
    $\forall \sn{\tau}[\vartheta] \in \Theta: \ \exists \vartheta' :
    \morph_{\sn{N}}(\sn{\tau})[\vartheta'] \in \Theta \
    \wedge
    \forall \on{N}\in \Nets : 
    \vartheta'(\on{N}) = \morph_{\on{N}}(\vartheta(\on{N}))$ 
  \end{enumerate}
\end{definition}

Observe that the second condition in Definition
\ref{def:eos-automorph} implicitly establishes a permutation within
the events in $\Theta$.

We naturally extend the concept of \EOS{}-automorphism to \emph{marked} \EOS{}.
\begin{definition}
  \label{def:markedeos-automorph}
  An automorphism between
  \((\mathit{OS}, \mu)
  \) and \((\mathit{OS}, \mu')  
  \) is an \EOS{-}automorphism $\morph_{\mathit{OS}}$
  such that:
  \\  
  If $\mu = \sum_{k=1}^{|\mu|} \sn{p}_k[\on{M}_k]$ then $\mu' = \morph_{\mathit{OS}}(\mu) := \sum_{k=1}^{|\mu|} \morph_{\sn{N}}(\sn{p}_k)\big[\morph_{d(\sn{p}_k)}(\on{M}_k)\big]$.

  The markings $\mu_1$ and $\mu_2$ of $\mathit{OS}$ are said to be \key{equivalent} if and only if there exists $\morph_{OS} :  (\mathit{OS}, \mu_1) \rightarrow (\mathit{OS}, \mu_2)$. We write (abusing notation) $\mu_1 \cong \mu_2$.
  
  Again, $\cong$ defines an equivalence relation on the set of markings of $\mathit{OS}$.
\end{definition}

The analogy of p/t automorphism and \EOS{}-automorphism also holds for
the \EOS{} firing rule (cf. Def. \ref{def:eos-fire}). Observe that an
automorphism $\morph_{\mathit{OS}}$ maintains the firing predicate
(\ref{eq:firepredicate}).

\begin{lemma}
  The firing predicate $\phi$ is invariant w.r.t.~an  automorphism $\morph_{\mathit{OS}}$: 

  \begin{equation}
    \label{eq:fireeqmorf}  
    \begin{array}{rl}
      \phi(\sn{\tau}[\vartheta], \lambda, \rho)
      &\; \;\impl \; \;  \phi(\morph_{\mathit{OS}}(\sn{\tau}[\vartheta]), \morph_{\mathit{OS}}(\lambda), \morph_{\mathit{OS}}(\rho))
    \end{array}
  \end{equation}
  Here, $\morph_{\mathit{OS}}(.)$ is the congruent homomorphic application of $\morph_{\mathit{OS}}$ components to \EOS{} events and markings.
\end{lemma}
\begin{proof}
  We check the properties of $\phi$ as defined by (\ref{eq:firepredicate}). 
  By assumption we have $ \Pi^1(    \lambda) = \pre(\sn{\tau})$.
  This implies
  \(
  \Pi^1(\morph_{\mathit{OS}}(\lambda))
  \cong
  \Pi^1(\lambda)
  = \pre(\sn{\tau}) 
  \cong
  \pre(\morph_{\mathit{OS}}(\sn{\tau})) 
  \).
  Analogously, we obtain
  $    
  \Pi^1(\morph_{\mathit{OS}}(\rho)) = \post(\morph_{\mathit{OS}}(\sn{\tau}))$.
  For the object-nets we have for the preset:
  \[
    \Pi^2_{\on{N}}(\morph_{\mathit{OS}}(\lambda))
    \cong
    \Pi^2_{\on{N}}(\lambda)
    \geq 
    \pre_{\on{N}}(\vartheta(\on{N}))
    \cong
    \pre_{\on{N}}(\morph_{\mathit{OS}}(\vartheta(\on{N})))
  \]
  and    for the generated net-tokens in the postset we have:
  \[
    \begin{array}{rc lll}
      \Pi^2_{\on{N}}(\morph_{\mathit{OS}}(\rho))  
      &\cong&
              \Pi^2_{\on{N}}(\rho)  &&
      \\
      &=&
          \Pi^2_{\on{N}}(\lambda) 
                                    &{} - \pre_{\on{N}}(\vartheta(\on{N})) 
      &{} + \post_{\on{N}}(\vartheta(\on{N})) 
      \\
      &\cong&\Pi^2_{\on{N}}(\morph_{\mathit{OS}}(\lambda)) 
                                    &{}- \pre_{\on{N}}(\morph_{\mathit{OS}}(\vartheta(\on{N})) )
      &{}+ \post_{\on{N}}(\morph_{\mathit{OS}}(\vartheta(\on{N}))) 
        \quad\hfill\qed   
    \end{array}   
  \]
\end{proof}

From this Lemma
we immediately obtain the invariance of the firing rule.
\begin{proposition}
  The firing rule is invariant w.r.t. an  automorphism $\morph_{\mathit{OS}}$: 
  
  \begin{equation}
    \label{eq:firepredmorf} 
    \begin{array}{rl}
      \mu
      \labtra{\sn{\tau}[\vartheta](\lambda, \rho)}{\mathit{OS}} 
      \mu'  
      &\, \;\impl\,\;
        \morph_{\mathit{OS}}(\mu)
        \labtra{\morph_{\mathit{OS}}(\sn{\tau}[\vartheta])(\morph_{\mathit{OS}}(\lambda), \morph_{\mathit{OS}}(\rho))}{\mathit{OS}} 
        \morph_{\mathit{OS}}(\mu')  
    \end{array}
  \end{equation}
\end{proposition}
\begin{proof}
  We have to check the firing rule of Def.~\ref{def:eos-fire}.
  The event 
  $\morph_{\mathit{OS}}(\sn{\tau}[\vartheta])$
  is enabled in the mode
  $(\morph_{\mathit{OS}}(\lambda), \morph_{\mathit{OS}}(\rho))$
  since
  $
  \morph_{\mathit{OS}}(\lambda) \cong
  \lambda \sqsubseteq \mu
  \cong \morph_{\mathit{OS}}(\mu))$
  and
  $\phi(\morph_{\mathit{OS}}(\sn{\tau}[\vartheta]), \morph_{\mathit{OS}}(\lambda), \morph_{\mathit{OS}}(\rho))
  $
  holds, too.
  The    successor marking 
  is 
  \(
  \morph_{\mathit{OS}}(\mu')
  =
  \morph_{\mathit{OS}}(\mu)
  -\morph_{\mathit{OS}}(\lambda)
  + \morph_{\mathit{OS}}(\rho))
  \cong
  \mu - \lambda +  \rho
  =
  \mu'
  \).\hfill\qed
\end{proof}

\begin{example}
  In the following we give an example to illustrate our approach.
  Our \EOS{} models a kitchen 
  with two working stations $S_1$ and $S_2$, i.e., places
  where different activities can be executed.
  The whole scenario can be seen as a metaphor for flexible manufacturing systems, where the kitchen is the plant and the recipe is the workflow.

  We have a cook (corresponding to a robot) executing a simple recipe:
  First split eggs (action $a$);
  then, independently mix egg yolks with sugar (action $b$)
  and mix egg white with wine (action $c$);
  and, finally fill the white creme on top of pudding (action $d$).
  In process algebraic notation
  the recipe is denoted as $a; (b \| c); d$, where
  $\_;\_$ denotes sequential composition
  and $\_\|\_$ denotes an and-split (parallel execution).
  %
  %

  The corresponding \EOS{} is shown in Figure~\ref{fig:kitchen}.
  The system net has two places $S_1$ and $S_2$ for the kitchen stations.
  Each station has side transitions to execute those actions that are possible at each station.
  The cook/recipe can move freely between the two stations.
  We have only one object net that models the recipe.
  In the initial marking the recipe starts at station $S_1$.
  The symmetry of the locations
  is captured by the automorphism:
  \(
  \morph_{\mathit{OS}}(S_1) = S_2
  \),
  \(
  \morph_{\mathit{OS}}(p_1) = p_2
  \),
  and 
  \(
  \morph_{\mathit{OS}}(p_3) = p_4
  \).
  
  \begin{figure}[htbp]  \centerline{\includegraphics[width=0.8\textwidth]{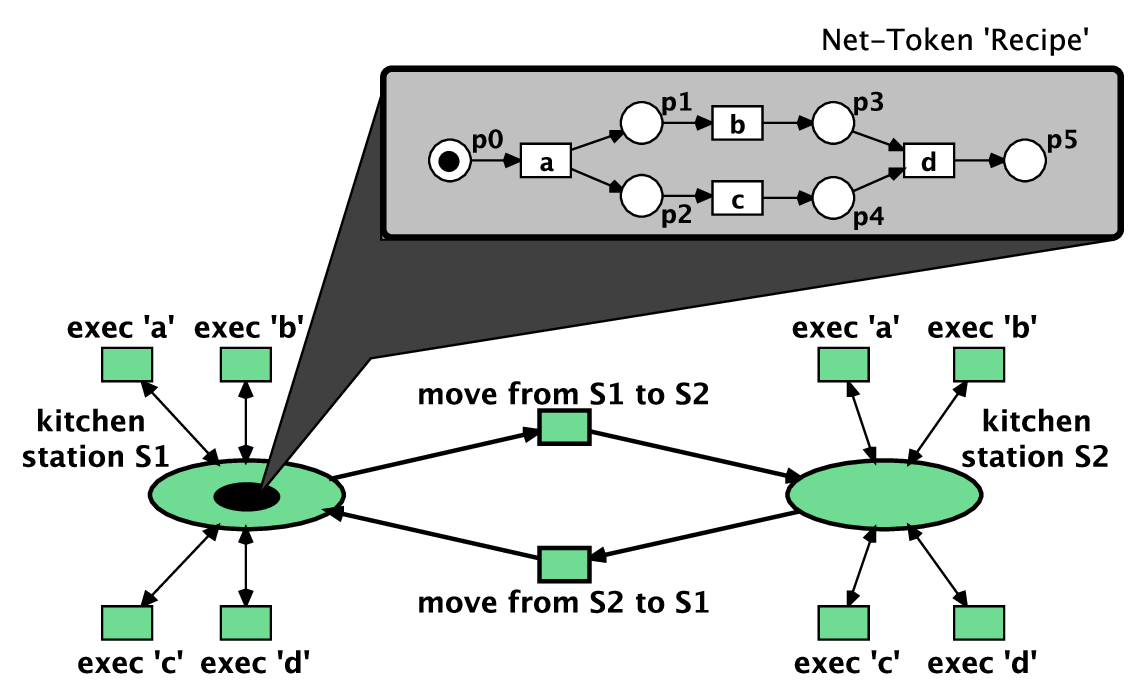}}
    \caption{\label{fig:kitchen} \EOS: A Kitchen with a Recipe as a Net-Token}
  \end{figure}%

  The resulting state space mainly describes the following execution:
  The recipe moves from kitchen station $S_1$
  where $a$ is possible to station $S_2$, where it is possible to execute both $b$
  and $c$ (in any order).
  Finally, the recipe moves back to 
  station $S_1$
  to execute  $d$.
  In between of all these steps the recipe may moves `erratically' between the two locations without making any progress to the recipe.

\end{example}

\subsection{Canonical Representations of \EOS{}}

Our basic idea is to provide some pseudo symbolic state representation based on automorphism.
We use canonical representatives, i.e.,
a most minimal representation
of an equivalence class (corresponding to an automorphism group) considering a total order on nodes: e.g., we assume that for any $p/t$ net $N$ there are two bijections $P_N \to \{1,\ldots,|P_N|\}$ and $ T_N \to \{1,\ldots,|T_N|\}$.
We derive the \EOS{}-automorphism
from the canonical representative of \key{p/t} components (net-tokens and system-net).

\begin{example}    
  Let us illustrate the main concepts using Figure \ref{fig:kitchen} (here $\sn{p}_1$ means $S_1$,
  $\sn{p}_2$ means $S_2$).
  The following markings:

  \[
    \mu  = \sn{p}_1[p_1 + p_4] + \sn{p}_1[p_2 + p_3] 
    + \sn{p}_2[p_2 + p_3]
  \]
  \[
    \mu'  = \sn{p}_1[p_2 + p_3] +
    \sn{p}_2[p_1 + p_4] + \sn{p}_2[p_2 + p_3]
  \]
  \noindent are actually automorphic, their canonical representative is (we use weights for clarity)
  \[
    \mu_c  = \sn{p}_1[p_1 + p_4] + 2\cdot \sn{p}_2[p_1 + p_4] 
  \]

  In this example, the state aggregation obtained through \EOS{} automorphism and canonization grows exponentially with the quantity of net-tokens.
\end{example}

We integrate the formalization of 
canonization for \EOS{}
along the results for Rewritable Petri Nets as given in  previous work \cite{Capra+2024-tcs,Capra:RP22}.
There we used the well established 
\texttt{Maude} tool \cite{all-about-maude} to define canonical representation of \emph{Rewritable} \key{p/t} Nets.
This formalization utilizes a multiset-based algebraic definition of mutable \key{p/t} nets. Whereas the method detailed in \cite{Capra:RP22} is generic (similar to other graph canonization techniques), the one outlined in \cite{Capra+2024-tcs} is specific to modular \key{p/t} nets created and modified using selected composition operators and a structured node labelling, making it significantly more efficient for symmetric models.

The two approaches can be seamlessly incorporated into the \EOS{} formalism within \texttt{Maude} \cite{enase23}, which uses a consistent representation of \key{p/t} terms.

In theory, it is known that generating
canonical representative is 
at least as complex as the graph automorphism problem (quasi polynomial).
We hope for a more efficient 
canonization since
for \EOS{} the \key{p/t} structure doesn't change: the canonical representative of a marking must retain the net structure, which
is a major source for  efficiency.
%
While the exact effect on efficiency is still being studied, the application of the method \cite{Capra+2024-tcs} appears to be promising, when \EOS{} components are built modularly.

Observe that the proposed canonization approach is pseudo-symbolic, because
from a canonical representative, we might reach equivalent markings through equivalent instances.
To approach entirely symbolic, canonical representative of firing modes should be calculated (using a uniform technique).

We finally like to mention that  in \cite{Koehler14-fi-lam}
we have considered another `in-built' symmetry of \EOS{}.
The nested multisets
$\alpha$ and $\beta$ that coincide in their projections give rise to
the so-called
\key{projection equivalence}:
\begin{equation}
  \label{eq:projection-equiv}
  \begin{array}{rcl}
    \alpha \cong \beta
    & :\iff &
               \Pi^1(\alpha) = \Pi^1(\beta) \land 
               \forall\on{N} \in \Nets: 
               \Pi^2_{\on{N}}(\alpha) = \Pi^2_{\on{N}}(\beta)          
  \end{array}
\end{equation}

\begin{lemma}[\cite{Koehler14-fi-lam}, Lem.~3.1]
  \label{lem:lem1-1}
  The enabling predicate is invariant with respect to 
  projection:
  \(
  \phi(\sn{\tau}[\vartheta], \lambda, \rho)
  \iff 
  (
  \forall \lambda', \rho':
  \lambda' \cong \lambda \land
  \rho' \cong \rho \impl
  \phi(\sn{\tau}[\vartheta], \lambda', \rho')
  )
  \)
\end{lemma}

This equivalence will give rise to additional
state space reductions.

\section{Conclusion}

In this paper, we have 
extended the concepts of
automorphisms
for Nets-within-Nets, here: \EOS{}.
In combination
with canonical forms for representing the markings we have obtained a potential significant reduction of
the size of the state space. We have integrated the approach in a \texttt{Maude} encoding of \EOS{}. 

In future work, we plan to investigate the degree of state space reduction for a broader set of case studies.
We further hypothesize that the automorphism concept for \EOS{} can be readily extended to Object-nets with any level of net and channel nesting, allowing net-tokens to flow between different levels.
We also aim to explore the practicality of less restrictive partial symmetry concepts.

\bibliographystyle{splncs04}

\bibliography{mk/defs,mk/koehler,mk/agent,mk/new,mk/eigenes,mk/konferenzen,lc}

\clearpage
\appendix

\end{document}